\pdfoutput=1
\documentclass[12pt]{article}%
\usepackage[utf8]{inputenc}%
\usepackage[tracking=true, letterspace=50, expansion=false]{microtype}

\usepackage{amsmath} %
\usepackage{amsthm} %
\usepackage{amssymb} 

\usepackage[margin=1in]{geometry} %
\usepackage{setspace} %
\usepackage[title]{appendix} %
\usepackage{booktabs} %

\usepackage{enumitem}

\newtheorem{lemma}{Lemma}[section]

\theoremstyle{definition}%
\newtheorem{remark}{Remark}

\DeclareMathOperator{\cv}{cv} %
\usepackage{mathtools}
\DeclarePairedDelimiter\abs{\lvert}{\rvert}
\DeclarePairedDelimiter\norm{\lVert}{\rVert}
\DeclarePairedDelimiter\indicatorfence{\{}{\}}
\DeclarePairedDelimiter\hor{[}{)}

\newcommand\1{\operatorname{I}\indicatorfence}

\usepackage{xcolor}%
\definecolor{webbrown}{rgb}{.6,0,0}

\usepackage[nolist]{acronym}
\begin{acronym}
  \acro{CI}{confidence interval}
  \acro{DOB}{date of birth}
  \acro{MSE}{mean squared error}
  \acro{RD}{regression discontinuity}
  \acro{OLS}{ordinary least squares}
  \acro{EHW}{Eicker-Huber-White}
  \acro{ATE}{average treatment effect}
  \acro{RBC}{robust bias correction}
  \acro{ROT}{rule of thumb}
  \acrodefplural{ROT}{rules of thumb}
\end{acronym}
\usepackage{chngcntr} 
\usepackage{etoolbox} 
\AtBeginEnvironment{appendices}{%
  \crefalias{section}{appendix}
  \counterwithin{figure}{section}
  \counterwithin{table}{section}
}

\usepackage[capposition=bottom]{floatrow}
\usepackage{tikz} 

\usepackage[
 bibencoding=utf8,%
 maxbibnames=3, %
 minbibnames=1, %
 backend=biber,%
 sortlocale=en_US,%
 style=apa,%
 doi=true,%
 url=true,
 eprint=true%
]{biblatex}

\usepackage{hyperref}%
\hypersetup{%
  breaklinks = true,%
  colorlinks = true,%
  anchorcolor = webbrown,%
  citecolor = webbrown,%
  filecolor = webbrown,%
  linkcolor = webbrown,%
  menucolor = webbrown,%
  urlcolor= webbrown,%
  citebordercolor= 1 0 0,%
  menubordercolor=1 0 0,%
  urlbordercolor=1 0 0,%
  runbordercolor=1 0 0,%
  pdftitle=When Can We Ignore Measurement Error in the Running Variable?,%
  pdfauthor=Yingying Dong \& Michal Kolesár}
\usepackage{cleveref}
\Crefname{equation}{Eq.}{Eqs.}
\crefname{assumption}{assumption}{assumptions}
\crefname{remark}{remark}{remarks}
\crefname{appsec}{Appendix}{Appendices}
\crefname{appsubsec}{Appendix}{Appendices}

\title{When Can We Ignore Measurement Error in the Running Variable?\thanks{We
    thank Tim Armstrong, Xiaohong Chen, and Christoph Rothe, conference participants at the ASSA 2023 Annual Meeting, and seminar participants at the University of Toronto for helpful comments. Kolesár acknowledges support by the Sloan Research Fellowship and by the National Science Foundation Grant SES-22049356. All errors are our own.}}
\author{Yingying Dong\thanks{email: \texttt{yyd@uci.edu}}\\
  University of California Irvine \and
  Michal Kolesár\thanks{email: \texttt{mkolesar@princeton.edu}}\\
  Princeton University}%
\date{\today}

\begin{document}

\maketitle

\begin{abstract}
  In many applications of regression discontinuity designs, the running variable used by the administrator to assign treatment is only observed with error. We show that, provided the observed running variable (i) correctly classifies the treatment assignment, and (ii) affects the conditional means of the potential outcomes smoothly, ignoring the measurement error nonetheless yields an estimate with a causal interpretation: the average treatment effect for units whose observed running variable equals to the cutoff. We show that, possibly after doughnut trimming, these assumptions accommodate a variety of settings where support of the measurement error is not too wide. We propose to conduct inference using bias-aware methods, which remain valid even when discreteness or irregular support in the observed running variable may lead to partial identification. We illustrate the results for both sharp and fuzzy designs in an empirical application.
\end{abstract}

\clearpage

\section{Introduction}

The key characteristic of \ac{RD} designs is that assignment of units to treatment is determined by whether the value of a particular covariate $X^{*}$, called the running variable, exceeds a fixed threshold $c$. Under weak continuity conditions, comparing units on either side of the threshold identifies the \ac{ATE} for units with $X^{*}=c$. However, in many cases, researchers only observe a noisy version $X$ of the true running variable $X^*$. In our survey of papers published in leading economics journals that featured \ac{RD} designs, 23\% of them used a running variable measured with error.\footnote{See \cref{sec:survey} for details.} Most commonly, the noise arose due to rounding or grouping, such as when researchers observe age in years, income reported in income brackets, or job tenure in months, while the administrator assigning treatment uses the exact birthdate, income, or job tenure.

The prevalence of this problem has given rise to a growing literature on measurement error in \ac{RD} designs (see, among others, \cite{HuKl10,dong15,BaLiWa16,PeSh17,DaLeBa17,BaBrDi20}; or
\cite{DiBaBr20}).
As this literature points out, ignoring the measurement error can lead to inconsistent estimates of the usual \ac{RD} estimand, the \ac{ATE} for units with $X^{*}=c$.\footnote{\label{fn:battistin}As we discuss in more detail in \Cref{sec:extension-fuzzy-rd}, an exception is \textcite{bbrw09},
  who, in the context of a fuzzy \ac{RD} design, consider measurement error with a point mass at zero, and give conditions under which ignoring the
  measurement error nonetheless yields a consistent estimate of a fuzzy \ac{RD} analog of this estimand.} This motivated the development of a variety of alternative estimation and inference procedures to recover this \ac{ATE}; the
solutions depend on the particular auxiliary assumptions about the form of the measurement error, or the availability of auxiliary datasets.

This paper makes the simple point that under easily interpretable conditions, existing \ac{RD} techniques provide inference for a slightly different estimand, the \ac{ATE} for units with $X=c$. That is, one could interpret the \ac{RD} analysis as if $X$ were the true running variable. For instance, suppose $X^{*}$ is birthdate, but we only observe the year of birth $X$. While standard \ac{RD} analysis does not yield valid inference for the \ac{RD} estimand associated with $X^*$, the \ac{ATE} for individuals born on the cutoff date, it does provide valid inference for the \ac{ATE} for individuals born in the cutoff year.\footnote{Equivalently, this estimand is a weighted average of \acp{ATE} conditional on $X^{*}$, with weights given by the density of the measurement error conditional on $X=c$. In the example, it corresponds to a weighted average of \acp{ATE} for individuals born on each day of the cutoff year, weighted by the relative birthdate frequencies.}

This result relies on two key conditions: (i) using $X$ as a running variable correctly classifies the treatment assignment, and (ii) the conditional means of the potential outcomes are smooth in $X$. The first condition holds automatically for certain types of rounding error; more generally it requires removing observations at or in the immediate vicinity of the threshold, resulting in a ``doughnut design''. Our approach is thus most relevant when the support of the measurement error is relatively narrow (which includes most settings with rounding or grouping error); otherwise the doughnut trimming may end up removing too many observations and preclude informative inference.

For the second condition, we give a formal result showing it holds under \emph{weaker} conditions than those needed for inference on the \ac{ATE} conditional on $X^*=c$, if we were to observe $X^*$: intuitively, the measurement error smooths out kinks or other irregularities in the conditional mean of the outcome given $X^*$. Inference can be conducted using bias-aware inference methods \parencite[e.g.][]{ArKo18optimal,ArKo20,NoRo21}, which automatically adapt to the potentially irregular support of $X$. In particular, we can ignore the measurement error in the sense that estimation and inference on the trimmed data can proceed as if $X$ were the running variable used by the administrator to assign treatment.

An appealing feature of focusing on the \ac{ATE} for units with $X=c$ is that valid inference relies only on assumptions about the smoothness of the conditional mean of the potential outcomes given $%
X$. Such assumptions are easy to interpret, and are (partially) testable. Furthermore, inference is standard in that we can directly apply existing methods. In contrast, for inference on the average treatment effect for units with $X^{*}=c$, one needs to either make specific assumptions about latent objects, such as the distribution of the measurement error given $X$ or $X^{*}$, or have access to auxiliary data; furthermore, the form of the estimator depends on the exact form of these specific assumptions.

The rest of the paper proceeds as follows. \Cref{sec:setup-main-result} gives the setup and the main results, and discusses their applicability under different types of measurement error. \Cref{sec:empir-appl-making} illustrates the results in an empirical application. \Cref{sec:conclusion} concludes. Auxiliary results appear in the appendix.

\section{Setup and main results}\label{sec:setup-main-result}

We are interested in the effect of a treatment $T$ on an outcome $Y$. Let $Y(t)$ denote the potential outcomes, $t\in\{0,1\}$. The observed outcome is given by $Y=Y(0)+T(Y(1)-Y(0))$. An administrator assigns individuals to treatment if their running variable $X^{*}$ crosses a threshold $c$, which we normalize to $c=0$. Let $Z=\1{X^{*}\geq 0}$ denote the indicator for treatment assignment. In a sharp \ac{RD} design, all individuals comply with the treatment assignment, so that $T=Z$. Let $g^{*}_{t}(x):=E[Y(t)\mid X^{*}=x]$ denote the conditional mean of the potential outcomes given $X^{*}$, and let $g^{*}(x):=E[Y\mid X^{*}=x]=\1{x\geq 0}g^{*}_{1}(x)+\1{x <0}g^{*}_{0}(x)$ denote the conditional mean of the observed outcome.

If the conditional means $g^{*}_{t}$ are continuous at $0$, the jump in $g^{*}$ at $0$ identifies an average treatment effect for units at the threshold \parencite{htv01},
\begin{equation} \label{eq:taustar_srd}
\tau^{*}:=E[Y(1)-Y(0)\mid X^{*}=0]=\lim_{x\downarrow 0}g^{*}(x)-
\lim_{x\uparrow 0}g^{*}(x).
\end{equation}
If $X^{*}$ were observed, and we strengthen the continuity assumption by
placing nonparametric smoothness assumptions on $g^{*}$, several methods for
estimation and inference on $\tau^{*}$ are available
\parencite[see,
e.g.,][]{ImKa12,CaCaTi14,ArKo18optimal,ArKo20,ImWa19}. Under parametric
restrictions on $g^{*}$, we can leverage standard parametric regression
methods for estimation and inference.

We do not observe $X^{*}$ directly, however; instead we observe $X=X^{*}-e$,
where $e$ is measurement error.\footnote{Our setting is distinct from that in
  \textcite{eckles2020}, where the running variable $X$ observed by the
  researcher is also the variable used by the administrator, and it can be
  thought of as a noisy measure of some latent variable $X^{*}$ that affects
  potential outcomes.} As our leading example, we focus on the case where $e$
represents rounding or grouping error. By \emph{rounding}, we mean that the observed running variable can be written as $X=r(X^{*})$, where the rounding function $r$ is a monotone step function that is idempotent (rounding a number twice is the same as rounding it once, $r(r(X^{*}))=r(X^{*})$), and the steps
are equal-sized. For instance, if we observe a rounded-down version of $X^{*}$, then the rounding function corresponds to the floor function, $r(X^{*})=\lfloor X^{*}\rfloor$, and $e=X^{*}-\lfloor X^{*}\rfloor$.
Rounding is a special case of \emph{grouping}. Grouping allows $r(\cdot)$ to be any idempotent step function, not necessarily with equal-sized steps; it arises when the sample space for $X^{*}$ is partitioned into subsets, and $X=r(X^*)$ is a numeric value representing the subset \parencite{HeRu91}. For instance, under interval reporting, $X$ may correspond to the midpoint or one of the endpoints of the interval that $X^{*}$ belongs to. In spatial \ac{RD}, it is common to observe the centroid of the unit's ZIP code or county instead of its exact location.

Our framework, outlined in \Cref{sec:proposed-approach} for sharp \ac{RD}, and extended to fuzzy \ac{RD} settings in \Cref{sec:extension-fuzzy-rd}, only imposes high-level conditions on $e$ that allow the measurement error to take
many other forms besides grouping; $X$ may be discrete, continuous, or mixed.
\Cref{sec:cond-valid-prop} discusses sufficient low-level conditions in particular settings.
In addition to grouping, we consider classical measurement
error ($e$ is independent of $Y(1), Y(0)$ and $X^{*}$), Berkson measurement error ($e$ is independent of $Y(1), Y(0)$ and $X$), as well as \emph{heaping}.
While under grouping, the grouping function $r(X)$ is the same for each unit (so that the conditional distribution of $X$ given $X^{*}$ is degenerate), under heaping one may round the data according one of several rounding functions $r_{u}(\cdot)$, $u=1,2,\dotsc, K$, with different ranges; we observe $X=r_{U}(X^*)$, where $U$ is a random variable, possibly correlated with potential outcomes or $X^*$, determining which rounding function is used \parencite{HeRu91}. For instance, some parents may report ages of their children rounded to the nearest year or nearest half-year, while others may not round at all and report the exact date of birth \parencite[e.g.][]{HeRu90}.
The degree of rounding may depend on characteristics of the parents or the age of the child (ages of very young children and children of parents with more education are usually more likely to be precisely reported).
The density of the running variable then displays ``heaps'' at ages that are multiples of $0.5$.

\subsection{Estimation and inference with measurement error}\label{sec:proposed-approach}

Our approach is based on the observation that any variable $X$ can serve as a running variable, provided that it correctly classifies the treatment assignment, and provided that the conditional mean functions of the potential outcomes given $X$ are continuous at the cutoff:
\begin{enumerate}[label = (\emph{C\arabic*})]
\item\label{item:c1} $\1{X\geq 0}=Z$ almost surely.
\item\label{item:c2} $g_{t}(x):=E[Y(t)\mid X=x]$ is continuous at $0$ for $t=0, 1$.
\end{enumerate}
We give a detailed discussion of these conditions in our setting, where $X$ is a
mismeasured version of $X^{*}$, in \Cref{sec:cond-valid-prop} below.
Condition~\ref{item:c2} corresponds to the standard \ac{RD} continuity condition
from \textcite{htv01}, except it is applied to $X$ rather than $X^{*}$.

To interpret the estimand we consider, let $g^*_{t}(X^*, e):=E[Y(t)\mid X^*, e]$ denote the conditional mean of the potential outcome $Y(t)$, $t=0,1$, given \emph{both} the true value of the running variable and the measurement error, and let $\tau^*(X^{*}, e)=g^*_{1}(X^{*}, e)-g^*_{0}(X^{*}, e)=E[Y(1)-Y(0) \mid X^{*}, e]$ denote the conditional \ac{ATE}, conditional on both $X^*$ and $e$.
We denote the conditional \ac{ATE}, conditional on $X^*$ only, by $\tau^*(x)=g^*_{1}(x)-g^*_{0}(x)=E[Y(1)-Y(0) \mid X^{*}=x]$, so that $\tau^*=\tau^*(0)$.
While, as we discuss further below, our approach does allow the measurement error to affect potential outcomes, to clearly link the estimand we consider to the usual \ac{RD} estimand, it is useful to rule this possibility out and assume that the measurement error doesn't affect the conditional \ac{ATE}, at least when $X=0$ (note $E[Y(1)-Y(0) \mid X^{*}=x, X=0]=\tau^{*}(x, x)$):
\begin{enumerate}[label = (\emph{C\arabic*})]\setcounter{enumi}{2}
\item\label{item:c3} $\tau^*(x, x)=\tau^*(x)$.
\end{enumerate}
This condition is slightly weaker than the requirement that the measurement
error be non-differential, i.e., independent of $(Y(1), Y(0))$ given
$X^{*}$ \parencite[Chapter 2.6]{CaRuSt06}; for instance, rounding, grouping,
and classical measurement errors are all non-differential.
\begin{lemma}\label{theorem:sharp-designs}
  Suppose that $T=Z$ and conditions~\ref{item:c1} and \ref{item:c2} hold.
  Then the jump in the conditional mean function $g(x):=E[Y\mid X=x]$ at $0$
  identifies the \ac{ATE} for units with $X=0$,
  \begin{equation}\label{eq:tau_definition}
    \tau :=E[Y(1)-Y(0)\mid X=0] =\lim_{x\downarrow 0}g(x)-\lim_{x\uparrow 0}g(x).
  \end{equation}
  If, in addition, condition~\ref{item:c3} holds, then $\tau =\int \tau^{*}(e)dF_{e\mid X}(e\mid 0)$, where $F_{e\mid X}(e\mid x)$ is the conditional distribution of $e=X^{*}-X$ given $X=x$.
\end{lemma}
\begin{proof}
  Under condition~\ref{item:c1} $g(x)=\1{X\geq 0} g_{1}(x)+\1{X <0} g_{0}(x)$, so
  that
  $\lim_{x\downarrow 0}g(x)-\lim_{x\uparrow 0}g(x)=\lim_{x\downarrow
    0}g_{1}(x)-\lim_{x\uparrow 0}g_{0}(x)$, which equals $g_{1}(0)-g_{0}(0)$ by
  condition~\ref{item:c2}. By iterated
  expectations,
  $\tau =E\left[E\left[Y\left(1\right) -Y\left(0\right) \mid X^{*}, X=0\right]
    \mid X=0\right]$, which gives the second claim.
\end{proof}

By \Cref{theorem:sharp-designs}, we can effectively ``ignore'' the measurement
error in the observed running variable $X$, in that we can conduct the analysis
as if $X$ were the running variable used by the administrator to assign
treatment, provided that we align the target of inference accordingly, setting
it to $\tau$. Under condition~\ref{item:c3}, we can interpret $\tau$ as a
weighted average of \acp{ATE} conditional on $X^{*}=x$, $\tau^*(x)$; in contrast $\tau^*$ simply corresponds to the conditional \ac{ATE} at $0$, $\tau^*(0)$. If condition~\ref{item:c3} doesn't hold, then both estimands may be expressed as weighted averages of conditional \acp{ATE}, conditional on both $X^*$ and $e$, but the weighting functions are different, which makes them a little harder to compare. Specifically, by iterated expectations, $\tau =\int \tau^{*}(e, e) dF_{e\mid X}(e\mid 0)$, while $\tau^* =\int \tau^{*}(0, e)dF_{e\mid X^*}(e\mid 0)$, where $F_{e\mid X^*}$ is the conditional distribution of $e$ given $X^*$.

\begin{remark}[Comparison of $\tau$ and $\tau^{*}$]\label{remark:tau_comparison}
  The parameter $\tau^{*}$ corresponds to an \ac{ATE} for units with the latent running variable $X^{*}$ equal to $0$. If $X^{*}$ is the birthdate of an individual, for instance, then $\tau^{*}$ is the \ac{ATE} for those born on the cutoff date. If $X$ is month of birth, then $\tau$ corresponds to the \ac{ATE} for those born in the same month as the individuals born on the cutoff date. Since the measurement error is non-differential, $\tau$ can be expressed as a weighted average of $\tau^{*}(x)$, the \ac{ATE} for individuals born on day $x$, and the same month the individuals born on the cutoff date. The exact weights depend on the distribution of births in that month. If, for instance, birthdate is uniformly distributed within the month, then the weights are uniform. The estimands $\tau $ and $\tau^{*}$ are generally different unless $\tau^{*}(x)$ is constant on the support of $e$. For example, if $\tau^{*}(x)=a+bx$, and $E[e\mid X=0]=1/2$ (say when $e$ is uniform), which appears to be consistent with the results of our empirical application, then $\tau^{*}=a$, while $\tau =a+b/2$. Which parameter is more policy relevant depends on the particular policy counterfactual one has in mind.\footnote{One may object that the parameter $\tau$ is reverse-engineered in that \Cref{theorem:sharp-designs} shows $\tau$ is the parameter that our analysis happens to identify. The same criticism may be leveled at the result of \textcite{htv01} that \ac{RD} analysis in the absence of measurement error identifies $\tau^{*}$. We view both results as useful in separating the internal and external validity of the analysis.}
\end{remark}

For estimation and inference, we need to strengthen condition~\ref{item:c2} by
assuming that $g$ satisfies appropriate parametric or nonparametric smoothness
conditions. As a simple parametric approach, one may assume that $%
g(x)$ takes the form of a polynomial of degree $q$ on either side of the
threshold for values of $x$ within distance $h$ of the threshold. Then one could
estimate $\tau$ by a local polynomial regression of $Y$ onto
$m_{q}(X)=(\1{X\geq 0}, \1{X\geq 0}X, \dotsc, \1{X\geq 0}X^{q}, 1, \dotsc,
X^{q})$ (a polynomial in $X$ interacted with treatment assignment), using
\ac{OLS}.\footnote{While this covers a global approach by setting $h=\infty$, as
  discussed in \textcite{GeIm19}, such an approach may perform poorly relative to
  local approaches.} Specifically, given a sample $\{Y_{i}, X_{i}\}_{i=1}^{n}$,
the estimator is defined as
\begin{equation} \label{eq:local_polynomial}
\hat{\tau}_{h, q}^{Y} =
(1,0,\dotsc,0)^{\prime}\left(\sum_{i=1}^{n} \1{\abs{X_{i}} \leq h} m_{q}(X_{i})m_{q}(X_{i})^{\prime}\right)^{-1}
\sum_{i=1}^{n} \1{\abs{X_{i}}\leq h}m_{q}(X_{i})Y_{i}.
\end{equation}
Under i.i.d.\ sampling, inference can be conducted using \ac{EHW} standard
errors, provided that $X$ has at least $q+1$ support points on either side of
the threshold.

A limitation of the parametric approach is that if $g(x)$ is not exactly polynomial inside the estimation window, the estimator will be biased; consequently, \acp{CI} based on \ac{EHW} standard errors will undercover $\tau$. To address this issue, as the preferred inference method, we propose to use the bias-aware (or ``honest'') inference approach developed in \textcite{ArKo18optimal,ArKo20} and \textcite{KoRo18}. This approach enlarges the \acp{CI} by taking into account the potential finite sample bias of the estimator. In particular, letting $\hat{\sigma}(\hat{\tau}_{h, q}^{Y})$ denote the standard error, a \ac{CI} with level $1-\alpha$ takes the form
\begin{equation} \label{eq:bias_aware_CI} \hat{\tau}_{h, q}^{Y}\pm
  \cv_{\alpha}\left(B(\hat{\tau}_{h, q}^{Y})/ \hat{\sigma}(\hat{\tau}_{h,
      q}^{Y})\right) \cdot \hat{\sigma}(\hat{\tau}_{h, q}^{Y}).
\end{equation}
Here $\cv_{\alpha}(t)$ is the $1-\alpha$ quantile of a folded normal distribution $\abs{N(t,1)}$, and $B(\hat{\tau}_{h, q}^{Y})$ is a bound on the finite-sample conditional (on $X$) bias of the estimator. As a baseline assumption to bound the potential bias, we replace the parametric assumption that $g$ is polynomial with the weaker non-parametric assumption that $g\in \mathcal{F}_{RD}(M)$, where
\begin{equation}\label{eq:holder_smoothness}
  \mathcal{F}_{RD}(M)=\{f_{1}(x)\1{x\geq 0}+f_{0}(x)\1{x <0}\colon
  \norm{f^{\prime}_{0}}_{C^{1}}\leq M, \norm{f^{\prime}_{1}}_{C^{1}}\leq M\}.
\end{equation}
Here $\norm{f}_{C^{1}}=\sup_{x\neq x^{\prime}}\abs{f(x)-f(x')}/\abs{x-x'}$ is
the Lipschitz constant of $f$ (if $f$ is differentiable, then the constant is
the maximum of its derivative, $\norm{f}_{C^{1}}=\sup_{x}\abs{f'(x)}$). The
parameter space for $g$ thus corresponds to (the closure of) a family of functions
that are twice differentiable on either side of the cutoff, with the second
derivative bounded in absolute value by $M$, but are potentially discontinuous
at $0$. Under this assumption, it is optimal to run a local linear regression,
i.e.\ use the estimator $\hat{\tau}_{h, 1}^{Y}$. The (conditional on $X$) bias
of the estimator is maximized at the function
$h(x)=Mx^{2}(\1{x< 0}-\1{x\geq 0})/2$, so that $B(\hat{\tau}_{h, 1}^{Y})$ is
given by~\cref{eq:local_polynomial}, with $Y_{i}$ replaced by $h(X_{i})$. See
\textcite{ArKo20} and \textcite{KoRo18} for details. An appealing feature of the
bias-aware \ac{CI} is that because it accounts for the exact finite-sample bias
of the estimator, it is valid under any bandwidth sequence, including using a
fixed bandwidth; for example, the bandwidth $h$ may be selected to minimize the
(worst-case over $\mathcal{F}_{RD}(M)$) mean squared error, or the length of the
resulting confidence interval.\footnote{\label{fn:smoothness_constant}We implement this method in the
  empirical application in \Cref{sec:empir-appl-making}, where we also discuss
  the choice of $M$, the key tuning parameter. See \textcite{KoRo18},
  \textcite{ImWa19}, and \textcite{ArKo20} for a more detailed discussion,
  including a discussion of implementation issues.}

\subsection{Conditions for validity of proposed approach}\label{sec:cond-valid-prop}

Unlike existing approaches that seek to do inference on $\tau^{*}$ even in
presence of measurement error (e.g.\ \cite{HuKl10,dong15,DaLeBa17}; or
\cite{PeSh17}), we do not impose specific assumptions on the measurement error
distribution or require auxiliary data. Instead, our approach is based on the
observation that we can use existing parametric or nonparametric methods for
inference on $\tau $ provided that condition~\ref{item:c1} holds, and we
strengthen condition~\ref{item:c2} by assuming that the conditional mean
functions $g$ is smooth (in the sense that it is exactly polynomial inside the
estimation window, or else $g\in \mathcal{F}_{RD}(M)$).
These assumptions are high-level in that they are exactly the conditions needed to interpret the \ac{RD} design with the observed running variable as a valid \ac{RD} design.
In the following remarks, we discuss in detail sufficient conditions for these assumptions in specific measurement errors settings.
We also discuss related practical issues.

\begin{remark}[Correct classification of treatment
  assignment]\label{remark:donut}
  In a few special cases, such as when $X$ corresponds to $X^{*}$ rounded down
  to the nearest integer, and the threshold $c$ is an integer,
  condition~\ref{item:c1} holds
  automatically.\footnote{\label{fn:c1_example}Another example when
    condition~\ref{item:c1} holds in the full sample is when only values of $X$ outside the
    immediate vicinity of the cutoff are error-ridden. Specifically, let
    $X^{*}-U$ be an initial noisy measurement, where $U$ is pure measurement
    error with bounded support. We observe a follow-up exact measurement,
    $X=X^{*}$ with probability $p(X^{*}-U)$, and observe the noisy measurement
    otherwise. The probability $p(\cdot)$ varies smoothly and equals $1$ near
    the cutoff, when it is necessary to wait for the follow-up measurement to
    determine treatment assignment.} In general, however, measurement error in
  $X$ may induce misclassification of the treatment assignment for values of $X$
  equal to the cutoff or in its immediate vicinity. In such cases,
  condition~\ref{item:c1} requires dropping observations with such values of
  $X$, resulting in a ``doughnut'' design \parencite[e.g.][]{bglw11,adkw11}.

  The exact form of such doughnut trimming depends on the support of the
  measurement error. Under grouping error, we need to remove observations corresponding to the subset containing $0$. For instance, under ordinary rounding, or rounding up to the nearest integer, we need to remove observations with $X=0$. Under interval measurement, we need to remove $X$ that falls into the interval containing $0$. Under Berkson, classical, or other types of measurement error with bounded support $[s_{0}, s_{1}]$, with $s_{0}\leq 0\leq s_{1}$, we need to remove observations with $X\in \hor{-s_{1}, -s_{0}}$.

  A limitation of our approach is that it does not handle settings in which the
  support $e$ is unknown or unbounded. Furthermore, if the support of $e$ is
  wide, the doughnut trimming may result in removing many observations, and
  preclude informative inference. One way to proceed in such cases is focus on
  inference about $\tau^*$ under parametric assumptions about the measurement
  error distribution and the form of $g^*$, as in
  \textcite{HuKl10}.\footnote{Inference on $\tau^*$ without
    parametric restrictions on $g^*$ is challenging, because $\tau^*$ is
    generally unidentified unless the measurement error distribution is
    completely known. Furthermore, even if the distribution is known, the rates
    of convergence can be very slow. For example, if $e$ is classical Gaussian
    measurement error, the lower bounds in \textcite{FaTr93} suggest that the
    rate is logarithmic in the sample size.} If the researcher observes the
  treatment assignment $Z$, in addition to $Y$ and $X$, we can easily infer
  which units are misclassified. \textcite{PeSh17} discuss how to use this
  information to recover $\tau^*$ without parametric restrictions on $g^*$ or on
  the measurement error distribution.\footnote{\label{fn:daleba}In the context of fuzzy \ac{RD},
    \textcite{DaLeBa17} develop a nonparametric approach to estimation of
    $\tau^{*}$ that is likewise flexible about the form and support of the
    measurement error distribution, but they require the econometrician to
    observe both $X$ and $X^{*}$ in the subsample of treated individuals.}
\end{remark}

\begin{remark}[Smoothness of $g$]\label{remark:smoothness}
  Our proposed approach requires smoothness of $g$ in the sense that $g\in\mathcal{F}_{RD}(M)$. To discuss this condition, suppose that condition~\ref{item:c1} holds, so that we may write $g(x)=\1{x\geq 0}g_1(x)+\1{x <0}g_0(x)$, where, by iterated expectations,
  \begin{equation}\label{eq:gstar_x_e}
    g_{t}(x)=E[Y(t)\mid
    X=x]=E[g_{t}^{*}(X^{*}, e)\mid X=x].
  \end{equation}
  In \Cref{lemma:smoothness} in \Cref{sec:auxiliary-results}, we give a formal
  result showing that if (a) the conditional distribution of the measurement error $F_{e\mid X}(e\mid x)$ is smooth in $x$; and (b) the effect of the measurement error on potential outcomes is smooth, so that $g_{t}^{*}(X^{*}, e)$ is smooth in the second argument, then the condition $g\in\mathcal{F}_{RD}(M)$ is \emph{weaker} than the analogous smoothness requirement needed for inference on $\tau^{*}$ if $X^{*}$ were observed: the smoothness of $g_{t}$ in $X$ is \emph{greater} than the smoothness of $g_{t}(X^{*}, e)$ in $X^{*}$.
  Consequently, \cref{eq:holder_smoothness} will hold for $g$ even in settings where it may not hold for $g^{*}(x)=\1{x\geq 0}g^*_1(x)+\1{x <0}g^*_0(x)$.
  Condition (b) holds trivially for Berkson measurement error, since then $F_{e\mid X}(e\mid x)$ doesn't depend on $x$ at all; it also holds for classical measurement error if the density of $X^*$ is smooth.\footnote{Specifically, by Bayes' rule, the conditional density of $e$ given $X=x$ may be written as
  $f(e;x)=f_{X^*}(x+e)f_e(e)/\int f_{X^*}(x+e)f_e(e)de$, where $f_{X^*}$ and $f_e$ denote the density of $X^*$ and $e$, respectively. This expression is smooth in $x$ so long as the density of $X^*$ is smooth and the marginal density of $X$, $\int f_{X^*}(x+e)f_e(e)de$, is bounded away from zero.}

To gain intuition for this result, suppose first that the measurement error is
non-dif\-fer\-en\-tial, so that $g_{t}^{*}(X^{*}, e)=g_{t}^{*}(X^{*})$, and
condition (b) above holds trivially. Then we may write \cref{eq:gstar_x_e} as
$g_{t}(x)=E[g_{t}^{*}(X^{*})\mid X=x]$. \Cref{lemma:smoothness} then formalizes
the notion that the conditional expectation $E[\cdot\mid X=x]$ ``smooths out''
non-linearities in $g_{t}^{*}$.\footnote{See, for example, \textcite[Section
  4]{newey13aer} for a discussion in the context of nonparametric instrumental
  variables regression, where $X^{*}$ plays the role of an endogenous variable,
  and $X$ plays the role of an instrument. One consequence of this smoothing in
  the current context is that the measurement error may smooth out the
  discontinuity of $g^{*}$ at the cutoff, making $g$ continuous, and causing
  condition~\ref{item:c1} to fail unless we restrict the sample as discussed
  in~\Cref{remark:donut}.} For example, if $g_{t}^{*}$ contains kinks (so that
it has smoothness index $1$ and \cref{eq:holder_smoothness} fails for $g^{*}$),
these kinks will be smoothed out by the measurement error provided that the
conditional density of $e$ given $X$ is continuous with a bounded slope (in
which case $g_{t}$ will have smoothness index $2$, and
\cref{eq:holder_smoothness} will hold for $g$). This smoothing effect is
greatest under grouping error, or more generally whenever $X$ is discrete: under
the interpretation in \Cref{remark:discreteness} below, \cref{eq:holder_smoothness} always holds for appropriately chosen $M$ (so \Cref{lemma:smoothness} is not needed), since we
can always smoothly interpolate $g_{t}(x)$ through the support points of $X$,
using, say, spline interpolation \parencite[e.g.][]{spath95}.

The result in \Cref{lemma:smoothness} also goes through under differential measurement error, provided the error affects the potential outcomes smoothly---this is analogous to the result that $g_{t}^{*}(X^{*})$ remains to be smooth in $X^{*}$ even if agents can manipulate their running variable, so long as the manipulation is not perfect \parencite{lee08}.

  For the parametric approach, if condition~\ref{item:c1} holds, a sufficient
  condition for $g$ to be polynomial of degree $q$ on either side of the
  threshold is that $g_{0}^{*}(X^{*}, e)$ and $g_{1}^{*}(X^{*}, e)$ are
  multivariate polynomials of degree $q$, and $E[e^{j}\mid X]$ for
  $j=1,\dotsc, q$ are polynomials of degree $j$. This follows directly from the
  binomial theorem.
\end{remark}

While conditions (a) and (b) in \Cref{remark:smoothness} are relatively mild, they necessitate subpopulation analysis under heaping. For example, consider using birthweight as a
running variable to identify the effects of hospital care on infant health as in \textcite{adkw10}. Suppose that some (but not all) hospitals report rounded
rather than precise birthweight---then condition (a) fails since the
distribution $F_{e\mid X}(e\mid x)$ changes discontinuously as at round values
of $x$. Furthermore, suppose that hospitals with fewer resources are more likely
report rounded rather than precise birthweight, as argued in
\textcite{BaLiWa16}. Then condition (b) also fails, since $g^{*}_{t}(X^{*}, e)$
is then potentially discontinuous at $e=0$ due to different hospital composition
under rounded vs exact reporting. Suppose, however, that the conditional mean functions $g_t$ are smooth in the subpopulation of hospitals that don't round birthweight. Conditions (a) and (b) then hold if we drop the heaping points, only keeping observations with non-round values of $X$, as suggested by \textcite{BaLiWa16}. The estimand $\tau$ corresponds to a conditional \ac{ATE} for units with $X=X^{*}=0$ born in a hospital that reports exact birthweight.

In some instances of heaping, all values of $X$ are rounded, but the rounding precision differs across units.
Say some individuals round age to the nearest month,
while others round to the nearest year or half-year. To ensure conditions (a) and (b) hold in such cases, we drop individuals with age in months that is a multiple of $6$. The remaining sample then only contains those who report age in months. We can interpret the estimand as the \ac{ATE} for those born in the cutoff month in the subpopulation of individuals who report the running variable with the greatest precision.

Since the conditional mean function $g(x)$ is identified over the support of
$X$, smoothness assumptions such as $g\in \mathcal{F}_{RD}(M)$ are testable.
Problems such as heaping are often apparent from simple plots of undersmoothed
binned averages of the outcome against $X$ \parencite[see, e.g., Figure 1
in][]{bglw11}, and one can also conduct more formal specification tests
\parencite[see, e.g.,][Appendix S.3]{KoRo18}.

\begin{remark}[Irregular support of $X$]\label{remark:discreteness}
  The measurement error may result in a coarsening of the support of the
  observed running variable $X$ relative to $X^{*}$. Under grouping, for example, $X$ becomes discrete even if $X^{*}$ is continuously distributed.
  Furthermore, there may be a gap in the support around $0$ due to doughnut trimming (see \Cref{remark:donut}). In such ``irregular'' cases, since
  conditional mean functions are only well-defined over the support of the
  conditioning variable, following \textcite{KoRo18} and \textcite{ImWa19}, we
  interpret smoothness assumptions such as~\cref{eq:holder_smoothness} to mean
  that there exists a function $g(x)\in\mathcal{F}_{RD}(M)$ with domain
  $\mathbb{R}$ such that $E[Y\mid X]=g(X)$ with probability one. With discrete
  $X$ or under a doughnut design, there will be multiple functions $g$
  satisfying this condition, and the parameter $\tau$ will only be partially
  identified.

  An advantage of bias-aware inference is that the estimator and \ac{CI} construction remains the same whether the support of $X$ is continuous, discrete, or otherwise irregular, and whether $\tau$ is point or partially identified.
  Under irregular support of $X$, the finite-sample bias of the estimator may be large, but the \ac{CI} will automatically reflect it via a larger critical value (in such cases, the interval will converge to the identified set as the sample size $n\to\infty$).
  We illustrate these points in the empirical application in \Cref{sec:empir-appl-making}, where we show that under rounding error, confidence intervals for $\tau$ tend to be longer than confidence intervals for $\tau^{*}$ that one would obtain using the same construction if $X^{*}$ were observed.\footnote{\label{fn:ci_small_sample}In large samples, \acp{CI} for $\tau$ will be wider than the corresponding \acp{CI} for $\tau^{*}$ if $X^{*}$ were observed, because the former don't converge to a point, while the latter do. In finite samples, the variability of the estimators, which in general cannot be ranked, also matters, and the \ac{CI} for $\tau^*$ may end up being wider.}
\end{remark}

\subsection{Fuzzy designs}\label{sec:extension-fuzzy-rd}

In fuzzy \ac{RD} designs, only a subset of the individuals complies with the
treatment assignment, so that $T\neq Z$. In this case, \textcite{htv01} show
that the fuzzy \ac{RD} parameter can be interpreted as a local average
treatment effect for individuals who comply with the treatment assignment.
Let us reconsider their argument when we use a variable $X$ as the running
variable, not necessarily equal to the running variable $X^{*}$ used by the
administrator.\footnote{The original argument in \textcite{htv01} involved defining potential
treatments under counterfactual values of the running variable. However, the running variable may not be manipulable (e.g.\ when $X^{*}$ corresponds to a birthdate). We therefore use a slightly
different argument, based on manipulation of the treatment assignment. The
treatment assignment is typically manipulable, say by moving the cutoff.}

Let $T(1)$ denote the potential treatment status of the individual if they are
assigned to treatment, and let $T(0)$ denote their status if they are not
assigned to treatment. The observed treatment is given by $T=T(Z)$, and the
observed outcome is given by $Y=Y(T(Z))=Y(0)+T(Z)(Y(1)-Y(0))$. Let
$\mathfrak{C}$ denote the event that an individual is a complier, that is
$T(1)>T(0)$. Finally, in analogy to the conditional means $g$ and $g_{t}$, let
$p(x)=E[T\mid X=x]$ and $p_{z}(x)=E[T(z)\mid X=x]$ for $z\in \{0,1\}$.

We replace the sharp \ac{RD} condition that all individuals comply with the
treatment assignment ($T=Z$) with the weaker condition that a non-zero fraction
of individuals complies with it, and that nobody defies the treatment assignment
(in analogy with the monotonicity condition in \cite{ImAn94}):
\begin{enumerate}[label = (\emph{F\arabic*})]
\item\label{item:f1} $P(T(1)\geq T(0)\mid X=0) =1$, and $P(T(1)>T(0)\mid X=0)>0$.
\end{enumerate}
Next, we replace the continuity assumption~\ref{item:c2} with a continuity
assumption on the first stage and reduced form regression functions:\footnote{Analogous to an instrumental variables regression that uses $Z$ as an instrument, these are (non-parametric) regressions of $T$ and $Y$, respectively, onto $Z$ and $X$.}
\begin{enumerate}[label = (\emph{F\arabic*})]\setcounter{enumi}{1}
\item\label{item:f2} $p_{z}(x)$ and
  $E\left[Y\left(T\left(z\right) \right)\mid X=x\right] $, $z=0,1$, are
  continuous at 0.
\end{enumerate}
Intuitively, if treatment eligibility $Z$ did not change at the cutoff but was
instead fixed, this condition implies that the observed outcome $Y=Y(T(Z))$
would be continuous at $0$. As a result, any discontinuity must be due to change
in treatment eligibility, which allows for identification of causal effects.
Conditions~\ref{item:f1} and~\ref{item:f2} are analogous to the standard fuzzy \ac{RD}
assumptions, but applied to $X$ rather than $X^{*}$.

Finally, to link the estimand we consider to the usual \ac{RD} estimand, analogous to condition~\ref{item:c3}, it is useful to assume that the measurement error has no effect on the compliance probability
or the \ac{ATE} for compliers once we control for $X^{*}$.
\begin{enumerate}[label = (\emph{F\arabic*})]\setcounter{enumi}{2}
\item\label{item:f3}
  $P \left(\mathfrak{C}\mid X^{*}=x, X=0\right) =P \left(\mathfrak{C}\mid X^{*}=x\right) $ and
  $E\left[Y\left(1\right) -Y\left(0\right) \mid \mathfrak{C}, X^{*}=x, X=0\right]
  =E\left[Y\left(1\right) -Y\left(0\right) \mid \mathfrak{C}, X^{*}=x%
  \right] $.
\end{enumerate}
This is a slightly weaker requirement that the measurement error $e$ be
non-differential, i.e.\ independent of $(Y(1), Y(0), T(1), T(0))$ given $X^{*}$.
In analogy to condition~\ref{item:c3} in the sharp case, condition~\ref{item:f3}
is helpful for interpreting the estimand, but it is not necessary for validity
of our approach.

With this setup, we obtain a fuzzy \ac{RD} analog of
\Cref{theorem:sharp-designs}.
\begin{lemma}\label{theorem:fuzzy-designs}
  Suppose that conditions~\ref{item:c1}, \ref{item:f1}, and~\ref{item:f2} hold.
  Then
  \begin{equation*}
    \tau_{F} :=E\left[Y\left(1\right)
      -Y\left(0\right) \mid \mathfrak{C}, X=0\right] =\frac{\lim_{x\downarrow 0}g(x)-\lim_{x\uparrow 0}g(x)}{%
      \lim_{x\downarrow 0}p(x)-\lim_{x\uparrow 0}p(x)}.
  \end{equation*}
  If, in addition, condition~\ref{item:f3} holds, then
  $\tau_{F} =\int \tau_{F} ^{\ast}(e)\omega (e)dF_{e\mid X}(e\mid 0)$, where
  $\tau_{F} ^{\ast}(x):=E[Y(1)-Y(0)\mid \mathfrak{C}, X^{*}=x]$, and
  $\omega (e)=\frac{P(\mathfrak{C}\mid X^{*}=e)}{\int P(\mathfrak{C}\mid
    X^{*}=e)dF_{e\mid X}(e\mid 0)}$.
\end{lemma}%
\begin{proof}
    Observe that
    \begin{equation*}
        \begin{split}
          \lim_{x\downarrow 0}g(x)-\lim_{x\uparrow 0}g(x) & =\lim_{x\downarrow
            0}E[Y(T(1))\mid X=x]-\lim_{x\uparrow 0}E[Y(T(0))\mid X=x] \\
          & =E[Y(T(1))-Y(T(0))\mid X=0]\\
          & =E[(Y(1)-Y(0))(T(1)-T(0))\mid X=0] \\
          & =E[Y(1)-Y(0)\mid X=0,\mathfrak{C}]P(\mathfrak{C}\mid X=0),
        \end{split}%
    \end{equation*}%
    where the first equality uses the fact that $Y=Y(T(Z))$, and that by
    condition~\ref{item:c1}, $T=T(1)$ for individuals with $X\geq 0$, and
    $T=T(0)$ for those with $X<0$, the second equality uses
    condition~\ref{item:f2}, the third uses $Y(T(z))=Y(0)+T(z)(Y(1)-Y(0))$, and
    the last equality uses iterated expectations and
    condition~\ref{item:f1}. By analogous arguments,
    $\lim_{x\downarrow 0}p(x)-\lim_{x\uparrow 0}p(x)=P(\mathfrak{C}\mid
      X=0) $. The second claim follows by applying iterated expectations
    to the numerator and denominator of
    $\tau_{F}=\frac{E\left[(Y(1)-Y(0))\mathfrak{C} \mid X=0\right]}{E[\mathfrak{C}\mid
      X=0]}$, and using condition~\ref{item:f3}.
\end{proof}

Under perfect compliance, $T=Z$, \Cref{theorem:fuzzy-designs} reduces to
\Cref{theorem:sharp-designs}. In analogy to the sharp case, any variable
satisfying conditions~\ref{item:c1}, \ref{item:f1}, and~\ref{item:f2} can be
used as a running variable. With $X=X^{*}$, we obtain the standard result
that
    \begin{equation*}
        \tau_{F}^{\ast}:=\tau_{F} ^{\ast}(0)=\frac{\lim_{x\downarrow
        0}g^{*}(x)-\lim_{x\uparrow 0}g^{*}(x)}{\lim_{x\downarrow 0}p^{\ast
      }(x)-\lim_{x\uparrow 0}p^{\ast}(x)},
    \end{equation*}%
where $p^{\ast}(x)=E[T\mid X^{*}=x]$.

Unless the local average treatment effects $\tau_{F}^{\ast}(x)$ are constant on
the support of $e$, $\tau_{F} \neq \tau_{F}^{\ast}$. Since $\tau_{F}^{\ast}(x)$
is given by the ratio of the reduced form effect
$E[Y(T(1))-Y(T(0))\mid X^{*}=x]$ to the first stage effect
$E[T(1)-T(0)\mid X^{*}=x]$, whether $\tau_{F} ^{\ast}(x)$ is locally
constant depends on heterogeneity in both the reduced form and the first stage
conditional mean functions. Our empirical results in
\Cref{sec:empir-appl-making}, for example, are consistent with the reduced form
effect being approximately constant, while the first stage effect is
approximately linear, i.e., $E[Y(T(1))-Y(T(0))\mid X^{*}=x]\approx a$ and
$E[T(1)-T(0)\mid X^{*}=x]\approx b+cx$; further, the measurement error is approximately uniform on $[0,1] $. So $\tau_{F}\approx a/(b+c/2)$, while
$\tau_{F}^{\ast}\approx a/b$.

If the measurement error is differential, and condition~\ref{item:f3} doesn't
hold, then both $\tau^{*}_{F}$ and $\tau_{F}$ may be expressed as weighted
averages of conditional \acp{ATE} for compliers, conditional on both $X^*$ and
$e$, $\tau_{F} ^{\ast}(X^{*}, e) :=E[Y(1)-Y(0)\mid \mathfrak{C}, X^{*}, e]$.
Specifically, by iterated expectations,
$\tau_{F}= \frac{\int \tau_{F} ^{\ast}(e, e) p_{\mathfrak{C}}(e, e) dF_{e\mid
    X}(e\mid 0)}{\int p_{\mathfrak{C}}(e, e) dF_{e\mid X}(e\mid 0)}$, while
$\tau_{F}^{*}=\frac{\int \tau_{F}^{*}(0, e)p_{\mathfrak{C}}(0,e) dF_{e\mid
    X^{*}}(e\mid 0)}{\int p_{\mathfrak{C}}(0,e) dF_{e\mid X^{*}}(e\mid 0)}$,
where $p_{\mathfrak{C}}(X^{*}, e):=P(\mathfrak{C}\mid X^{*}, e)$.
\Cref{theorem:fuzzy-designs} is related to the result in \textcite{bbrw09} who
show that if we replace condition~\ref{item:c1} with the assumption that the
measurement error is non-differential, and has a point mass at
zero but is otherwise smooth, $\tau_{F}=\tau_{F}^{*}$. If the measurement error is differential, but affects
the potential outcomes and potential treatments smoothly, the arguments in \textcite{bbrw09} imply that
$\tau_{F}=\tau_{F}^{*}(0,0)$; this was shown in
\textcite{clpw15} in the context of fuzzy regression kink designs.

Similarly to the sharp case, if we assume that the conditional mean
functions $g(x)$ and $p(x)$ are polynomial inside a window $h$ of the
threshold, then we can estimate $\tau_{F}$ as a ratio of local polynomial
estimators
\begin{equation}\label{eq:fuzzy_RD_estimator}
\hat{\tau}_{h, q}=\hat{\tau}_{h, q}^{Y}/\hat{\tau}_{h, q}^{T},
\end{equation}
with $\hat{\tau}_{h, q}^{Y}$ defined in \cref{eq:local_polynomial}, and $\hat{\tau}_{h, q}^{T}$ defined analogously.\footnote{Equivalently, as noted in \textcite{htv01}, the estimator can be computed as a two-stage least squares estimator in a regression of $Y$ onto $T$ using $\1{X\geq 0}$ and instrument, and the remaining elements of $m_{q}(X)$ as exogenous covariates, using observations inside the estimation window.} If there are at least $q+1$ support points for $X$ on either side of the threshold and inside the estimation window,
then under i.i.d.\ sampling, standard errors for $\hat{\tau}_{h, q}$ can be
constructed based on the \ac{EHW} covariance matrix for $(\hat{\tau}_{h, q}^{Y}, \hat{\tau}_{h, q}^{T})$ using the delta method.

Our preferred approach weakens the polynomial assumptions on $g(x)$ and $%
p(x) $ by instead assuming that $g\in\mathcal{F}_{RD}(M_{y})$, and $p\in%
\mathcal{F}_{RD}(M_{t})$. While this assumption only delivers set identification if the support of $X$ is irregular (see \Cref{remark:discreteness}), we can use the bias-aware inference approach for constructing \acp{CI} that are asymptotically valid whether $\tau_{F}$ is point identified, set identified, or unidentified.\footnote{The parameter $\tau_{F}$ is unidentified if the instrument $\1{X\geq 0}$ is irrelevant in the sense that $P(T(1)>T(0)\mid X=0)=0$. Since the expression for $\tau_{F}$ in \Cref{theorem:fuzzy-designs} not well-defined in this case, one can define $\tau_{F}$ in an arbitrary way.} In particular, following \textcite{NoRo21}, we can test the hypothesis $H_{0}\colon\tau_{F}=\tau_{F, 0}$ by checking whether $0$ is in the bias-aware confidence interval based on $\hat{\tau}_{h,1}^{Y-\tau_{F, 0}T}$, and noting that the smoothness assumptions on $g$ and $p$ imply $%
E[Y-\tau_{F, 0}T\mid X=x]\in\mathcal{F}_{RD}(M_{y}+\abs{\tau_{F, 0}}M_{t})$. The confidence set for $\tau_{F}$ is constructed by collecting all values of $\tau_{F, 0}$ that are not rejected, similar to the construction of \textcite{AnRu49} confidence set in standard linear instrumental variables model.\footnote{\label{fn:smoothness_constant_fuzzy}We implement this method in our empirical application in \Cref{sec:empir-appl-making}, where we also discuss the choice of the smoothness constants $M_{t}$ and $M_{y}$. See \textcite{NoRo21} for a detailed discussion of implementation issues.}

\section{Empirical Application}\label{sec:empir-appl-making}

In this section, we use data from \textcite{HoHi16} to estimate the impact of
preregistration on youth turnout in an election. \textcite{HoHi16} leverage the
fact that in Florida, individuals who were ineligible to vote in the 2008
election (those born after November 4, 1990) were nonetheless eligible to preregister
to be added to the voter rolls for the next election. Those born before November 4,
1990 were already eligible to register regularly and vote in 2008. This
motivates a fuzzy \ac{RD} design, where the treatment $T$ is an indicator for
preregistering, the outcome $Y$ is an indicator for voting in the 2012 election, and the running variable $X^{*}$ is the proximity to the eligibility cutoff in days.

To illustrate the effects of measurement error in the running variable, we
compare this design to a fuzzy \ac{RD} design in which we (pretend to) only
observe individuals' month of birth, and hence use proximity to
November 1990 in months, $X$, as a running variable. We discard individuals born in November 1990, since their eligibility cannot be determined by month of birth alone (see \Cref{remark:donut}).
We show that, consistent with the discussion in \Cref{remark:tau_comparison,remark:discreteness}, (i) using proximity in days vs months yields different estimates, reflecting the impact of the rounding error on the estimand, and (ii) using month of birth generally leads to wider
\acp{CI}.

\begin{figure}[t]
\centering
\input{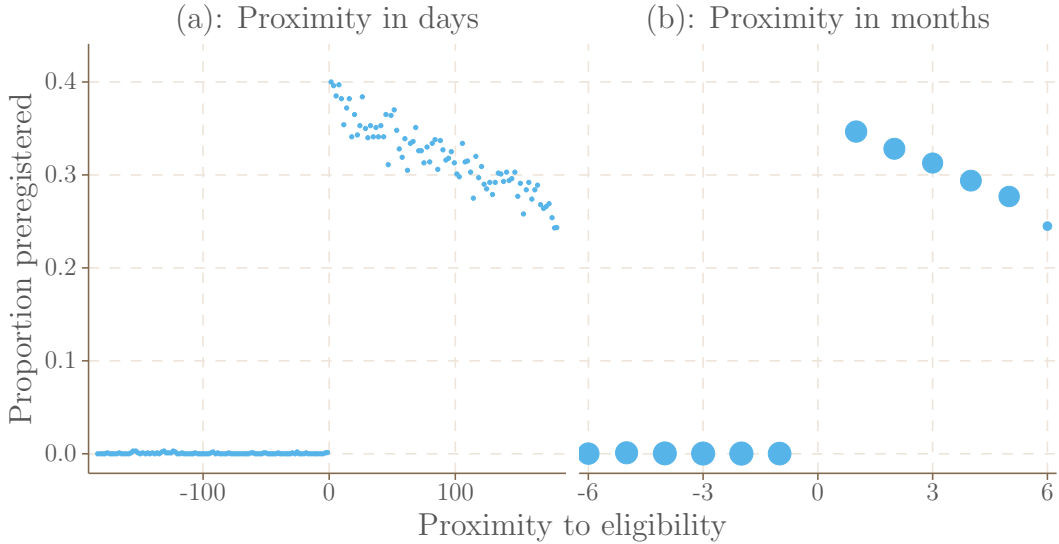}
\caption{Effect of proximity on preregistering.}\label{fig:first_stage_plot}
\floatfoot{\emph{Notes:} In panel (a), proximity is measured in days, and each point
  corresponds to an average of 1,000 individuals. In panel (b), proximity is
  measured in months, and each point corresponds to an average across all
  individuals born in a given month.}
\end{figure}

We first visualize both versions of the \ac{RD} design. In each case, the
sample size is 186,575, consisting of individuals born within 6 months of the
eligibility cutoff. \Cref{fig:first_stage_plot} presents the first stage,
plotting preregistration rate against proximity in days (panel (a)) or in months
(panel (b)). For ineligible individuals, the preregistration rate is essentially
0, while for eligible individuals, the preregistration rate is downward
slopping: those born further away from the cutoff preregister with lower
probability. There is a clear jump in the registration rate at the eligibility
threshold in either panel. \Cref{fig:reduced_form_plot} shows the reduced form, plotting the proportion who voted in the 2012 election against proximity to eligibility. In both panels, there is a small jump in the voting probability at the cutoff.

\begin{figure}[t]
\centering
\input{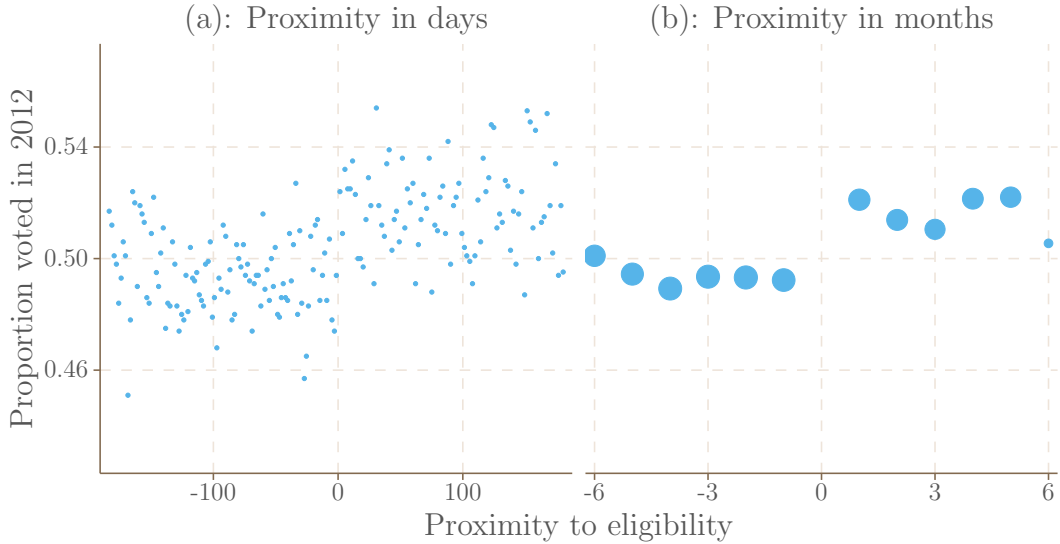}
\caption{Effect of proximity on voting.}\label{fig:reduced_form_plot}
\floatfoot{\emph{Notes:} In panel (a), proximity is measured in days, and each point
  corresponds to an average of 1,000 individuals. In panel (b), proximity is
  measured in months, and each point corresponds to an average across all
  individuals born in a given month.}
\end{figure}

We use five specifications to compute the fuzzy \ac{RD} estimator
in~\cref{eq:fuzzy_RD_estimator}, the sharp \ac{RD} estimators of the first stage and reduced form effects, and the associated confidence intervals.
For ease of comparison across specifications, all specifications use a uniform kernel and local linear regression ($q=1$).
The first specification follows \textcite{HoHi16}, and uses bandwidth set to $h=60$ days (or $h=2$ months), and the confidence
intervals to not account for the potential bias of the estimator. The second
specification differs only in that it uses a slightly larger bandwidth, $h=90$
days (or $h=3$ months).\footnote{These specifications can be interpreted as
  imposing a parametric linear functional form inside the estimation window.
  Alternatively, one can justify them by an ``undersmoothing'' argument: the
  specifications implicitly assume that the constants $M_{y}$ and $M_{t}$ are
  small enough so that the bias is negligible at these bandwidth choices.} The
third specification uses the \ac{RBC} method of \textcite{CaCaTi14}. For proximity
in days, we use the default ``MSE optimal'' bandwidth provided by their software
package; for proximity in months we use $h=3$.\footnote{The formal arguments
  justifying the \ac{RBC} method and the default bandwidth selector require the running
  variable to be continuous, which is not the case in either design.
  When proximity is measured in months, the discreteness causes implementation issues with the default ``MSE optimal'' bandwidth calculations.}

The last two methods implement the bias-aware approach. We use confidence
intervals given in \cref{eq:bias_aware_CI} for the first stage and reduced form
effects; for inference on the fuzzy \ac{RD} estimand, we use the
\textcite{NoRo21} construction. Implementing these methods requires a choice of
smoothness bounds for the first stage ($M_{t}$) and the reduced form ($M_{y}$). The results of \textcite{low97} and \textcite{ArKo18optimal} imply that picking $M_t$ and $M_y$ in a data-driven way without violating coverage requires further non-convex restrictions on the parameter spaces $\mathcal{F}_{RD}(M_{y})$ and $\mathcal{F}_{RD}(M_{t})$ for $g$ and $p$.\footnote{The problem of choosing the smoothness constants is essentially a non-parametric model selection problem. Echoing the difficulties with conducting valid post-model selection inference in parametric contexts \parencite[e.g.][]{LePo05}, \textcite{ArKo18optimal} show that bias-aware confidence intervals that assume the worst-case smoothness are in fact highly efficient at smooth functions. Thus, there is little scope for improvement by using data-driven choices of the smoothness constants.} A natural way of doing this is to relate the global smoothness of $g$ and $p$ to the local smoothness constants $M_y$ and $M_t$. We consider two ways of formalizing how the global and local smoothness relate. In particular, the fourth specification assumes, following the proposal in \textcite{ArKo20}, that $M_y$ is bounded by the smoothness of a global quartic approximation to $g$ on either side of the cutoff, as measured by the largest (in absolute value) second derivative of the fitted line; we impose an analogous assumption on $M_t$ and $p$. The fifth specification follows the suggestion in \textcite{ImWa19} to use a global quadratic regression instead, and, additionally, multiply the largest second
derivative of the fitted line by some moderate factor, taken here to be $2$.

Since there are many other reasonable ways of formalizing the idea that the local and global smoothness are related, we view these methods as merely \acp{ROT} for selecting the smoothness constants. To assess these rules, we use the visualization approach proposed in
\textcite{NoRo21}, described and implemented in \Cref{sec:empirical_extra}.
These visualizations suggest that the \textcite{ArKo20} \ac{ROT} is quite
conservative, and allows for $g$ and $p$ to be quite non-smooth. The second \ac{ROT} delivers more optimistic smoothness bounds that generate reasonably smooth conditional mean functions. To make the smoothness constants comparable across the specifications, we report the implied
smoothness constants after rescaling the running variable to have support
$[-1,1]$ (which amounts to multiplying the original smoothness constants by
$\max_{i}X_{i}^{2}$ and $\max_{i}(X_{i}^{*})^{2}$, respectively). Given a choice
of the smoothness constants, the bandwidth is selected so that the point
estimate defined in \cref{eq:local_polynomial} minimizes the worst-case (over
the chosen smoothness class) finite-sample MSE of the estimator.

\subsection{Results}\label{sec:results}

\begin{table}[t]
  \centering
  \caption{First stage estimates: effect of eligibility on preregistration.}\label{tab:fs}
  \small
\begin{tabular}{l c c c c c}
& \multicolumn{2}{c}{OLS} & \ac{RBC} & \multicolumn{2}{c}{Bias-aware inference}\\
\cmidrule(rl){2-3}\cmidrule(rl){4-4}\cmidrule(rl){5-6}
& (1) & (2) & (3) & (4) & (5)\\
\midrule
\multicolumn{3}{@{}l}{Panel A\@{}: Proximity in days}\\
Estimate & $0.384$ & $0.379$ & $0.393$ & $0.396$ & $0.381$\\
SE & $0.006$ & $0.005$ & $0.008$ & $0.009$ & $0.005$\\
95\% CI & $(0.373,0.396)$ & $(0.370,0.388)$ & $(0.378,0.409)$ & $(0.377,0.414)$ & $(0.371,0.391)$\\
Bandwidth & $60$ & $90$ & $43$ & $28$ & $86$\\
Eff. obs. & $63,220$ & $94,118$ & $43,538$ & $28,274$ & $89,776$\\
Rescaled $M_t$ & ${}$ & ${}$ & ${}$ & $1.015$ & $0.061$\\
\\[1ex]
\multicolumn{3}{@{}l}{Panel B\@{}: Proximity in months}\\
Estimate & $0.365$ & $0.363$ & $0.368$ & $0.365$ & $0.365$\\
SE & $0.009$ & $0.006$ & $0.017$ & $0.009$ & $0.009$\\
95\% CI & $(0.348,0.382)$ & $(0.351,0.375)$ & $(0.335,0.402)$ & $(0.303,0.428)$ & $(0.345,0.385)$\\
Bandwidth & $2$ & $3$ & $3$ & $2$ & $2$\\
Eff. obs. & $64,011$ & $94,662$ & $94,662$ & $64,011$ & $64,011$\\
Rescaled $M_t$ & ${}$ & ${}$ & ${}$ & $0.865$ & $0.092$\\
\\[-1ex]
\bottomrule
\end{tabular}
\floatfoot{\emph{Notes:} Column (1) uses local linear regression with bandwidth
  equal to 60 days (panel A) or 2 months (panel B), without any bias
  corrections. Column (2) is analogous, but uses bandwidth equal to 90 days or 3
  months. Column (3) uses the \ac{RBC} procedure, with the default ``MSE
  optimal'' bandwidth in panel A, and bandwidth equal to 3 months in panel B.
  Columns (4) and (5) report bias-aware confidence intervals, with bandwidth
  chosen to minimize the worst-case MSE\@. Column (4) uses the \ac{ROT} of
  \textcite{ArKo20} to choose the smoothness constant $M_{t}$, while column (5)
  uses the \ac{ROT} of \textcite{ImWa19}.\\Eff.\ obs refers to the number of
  observations inside the estimation window. The smoothness constants are
  reported after rescaling the running variable to have support $[-1,1]$.}
\end{table}

\Cref{tab:fs} presents the first stage estimates, the effect of the preregistration
eligibility on preregistration. The estimates are stable across the
specifications, in the range of 38--40\% when using proximity in days; the
estimates using proximity in months are slightly lower, in the range 36--37\%, but
still indicating a clear jump in the preregistration rate at the eligibility
threshold. This is consistent with our theory, discussed in
\Cref{remark:tau_comparison}, and reflects the difference between the parameters
$E[T(1)-T(0)\mid X^{*}=0]$ (the effect for those born on the cutoff date,
November 4) and $E[T(1)-T(0)\mid X=0]$ (the effect for those born in November,
the cutoff month). In particular, the latter estimand averages over individuals
born further away from the cutoff date, and \Cref{fig:first_stage_plot} suggests
that the preregistration probability is decreasing with the distance to the
cutoff. Since $X$ is discrete, the parameter $E[T(1)-T(0)\mid X=0]$ is not point
identified. The confidence intervals for the bias-aware specifications, which
account for this, are correspondingly wider than those in panel A, albeit they
still remain quite tight.

\begin{table}[t]
  \centering
  \caption{Reduced form estimates: effect of eligibility on voting.}\label{tab:rf}
  \small
\begin{tabular}{l c c c c c}
    & \multicolumn{2}{c}{OLS} & RBC & \multicolumn{2}{c}{Bias-aware inference}\\
    \cmidrule(rl){2-3}\cmidrule(rl){4-4}\cmidrule(rl){5-6}
    & (1) & (2) & (3) & (4) & (5)\\
\midrule
\multicolumn{3}{@{}l}{Panel A\@{}: Proximity in days}\\
Estimate & $0.028$ & $0.027$ & $0.032$ & $0.031$ & $0.028$\\
SE & $0.008$ & $0.007$ & $0.012$ & $0.012$ & $0.007$\\
95\% CI & $(0.012,0.044)$ & $(0.014,0.040)$ & $(0.008,0.056)$ & $(0.005,0.058)$ & $(0.013,0.044)$\\
Bandwidth & $60$ & $90$ & $38$ & $29$ & $83$\\
Eff. obs. & $63,220$ & $94,118$ & $39,195$ & $29,285$ & $86,881$\\
Rescaled $M_y$ & ${}$ & ${}$ & ${}$ & $1.401$ & $0.099$\\
\\[1ex]
\multicolumn{3}{@{}l}{Panel B\@{}: Proximity in months}\\
Estimate & $0.037$ & $0.034$ & $0.042$ & $0.037$ & $0.034$\\
SE & $0.013$ & $0.009$ & $0.025$ & $0.013$ & $0.009$\\
95\% CI & $(0.012,0.062)$ & $(0.017,0.051)$ & $(-0.007,0.090)$ & $(-0.085,0.159)$ & $(0.009,0.059)$\\
Bandwidth & $2$ & $3$ & $3$ & $2$ & $3$\\
Eff. obs. & $64,011$ & $94,662$ & $94,662$ & $64,011$ & $94,662$\\
Rescaled $M_y$ & ${}$ & ${}$ & ${}$ & $1.818$ & $0.121$\\
\\[-1ex]
\bottomrule
\end{tabular}
\floatfoot{\emph{Notes:} See \Cref{tab:fs}.}
\end{table}

\Cref{tab:rf} presents the reduced form estimates, the effect of preregistration
eligibility on voting. The point estimates are about 3\% for both designs and
stable across specifications. In line with the discussion in
\Cref{remark:discreteness}, the \acp{CI} are produced by the bias-aware
specifications are wider when using proximity in months, reflecting the loss of
point identification.

\begin{table}[t]
  \centering
  \caption{Fuzzy \ac{RD} estimates of the effect of preregistration on
    voting.}\label{tab:iv}
  \small
  \begin{tabular}{l c c c c c}
    & \multicolumn{2}{c}{OLS} & RBC & \multicolumn{2}{c}{Bias-aware inference}\\
    \cmidrule(rl){2-3}\cmidrule(rl){4-4}\cmidrule(rl){5-6}
    & (1) & (2) & (3) & (4) & (5)\\
    \midrule
\multicolumn{3}{@{}l}{Panel A\@{}: Proximity in days}\\
Estimate & $0.073$ & $0.072$ & $0.076$ & $0.080$ & $0.074$\\
SE & $0.021$ & $0.018$ & $0.032$ & $0.031$ & $0.018$\\
95\% CI & $(0.031,0.114)$ & $(0.037,0.106)$ & $(0.013,0.139)$ & $(0.012,0.143)$ & $(0.034,0.122)$\\
Bandwidth & $60$ & $90$ & $36$ & $29$ & $83$\\
Eff. obs. & $63,220$ & $94,118$ & $35,785$ & $29,285$ & $86,881$\\
Rescaled $M_y$ & ${}$ & ${}$ & ${}$ & $1.401$ & $0.099$\\
Rescaled $M_t$ & ${}$ & ${}$ & ${}$ & $1.015$ & $0.061$\\
\\[1ex]
\multicolumn{3}{@{}l}{Panel B\@{}: Proximity in months}\\
Estimate & $0.101$ & $0.094$ & $0.113$ & $0.101$ & $0.094$\\
SE & $0.035$ & $0.024$ & $0.068$ & $0.035$ & $0.024$\\
95\% CI & $(0.034,0.169)$ & $(0.047,0.141)$ & $(-0.020,0.246)$ & $(-0.268,0.505)$ & $(0.023,0.180)$\\
Bandwidth & $2$ & $3$ & $3$ & $2$ & $3$\\
Eff. obs. & $64,011$ & $94,662$ & $94,662$ & $64,011$ & $94,662$\\
Rescaled $M_y$ & ${}$ & ${}$ & ${}$ & $1.818$ & $0.121$\\
Rescaled $M_t$ & ${}$ & ${}$ & ${}$ & $0.865$ & $0.092$\\
\\[-1ex]
    \bottomrule
\end{tabular}
\floatfoot{\emph{Notes:} See \Cref{tab:fs}.}
\end{table}

\Cref{tab:iv} presents the fuzzy \ac{RD} estimates of the effect of
preregistration on voting. When eligibility is measured in months, the smaller
first stage estimates in panel B of \Cref{tab:fs} translate to larger estimates
of the effect of preregistration on voting, around 10\%, compared to 7--8\% when
eligibility is measured in days. When eligibility is measured in months, the
fuzzy \ac{RD} estimand, $\tau_{F}$, is the \ac{ATE} for compliers born in
November 1990, and thus averages over individuals born further away from the
cutoff than the estimand $\tau_{F}^{*}$ when eligibility is measured in days,
which corresponds to the \ac{ATE} for compliers born on November 4, 1990. If the
treatment effect for compliers born $x$ days from the eligibility threshold,
$\tau_{F} ^{\ast}(x)=E[Y(1)-Y(0)\mid \mathfrak{C}, X^{*}=x]$, is increasing
in $x$, then $\tau_{F}$ will be larger than $\tau_{F}^{*}$, which is consistent
with the results in \Cref{tab:iv}. However, the bias-aware confidence intervals
are fairly wide, and also consistent with $\tau_{F} ^{\ast}(x)$ being constant.

\section{Summary and conclusions}\label{sec:conclusion}

Measurement error is a common feature of \ac{RD} applications. We show that its presence does not have deleterious effects on the validity of existing inference methods, provided that one employs doughnut trimming to ensure that the observed running variable $X$ correctly classifies
the treatment assignment. Care needs to be taken when interpreting the estimand: it corresponds to the \ac{ATE} for units with the observed running variable $X$ equal to the cutoff,
rather than the usual parameter, the \ac{ATE} for units with the latent running
variable $X^{*}$ equal to the cutoff. We illustrate this point in an empirical
application.

\onehalfspacing
\phantomsection\addcontentsline{toc}{section}{References}

\printbibliography

\begin{appendices}
\crefalias{section}{appsec}
\crefalias{subsection}{appsubsec}
\section{Auxiliary results}
\subsection{Effect of measurement error on smoothness of conditional mean}\label{sec:auxiliary-results}

We now formalize the notion that measurement error smooths out non-linearities
in $g^{*}_{t}$. To this end, first we introduce some definitions. For a
real-valued function on a bounded set in $\mathbb{R}^{d}$, and a multi-index
$\alpha=(\alpha_{1}, \dotsc, \alpha_{d})$, let
$D^{\alpha}f=\partial^{\sum_{j=1}^{d}\alpha_{j}}f/\partial x_{1}^{\alpha_{1}}
\dotsb \partial x_{d}^{\alpha_{d}}$. For an integer $k$, let
$\norm{f}_{C^{k+1}}=\sum_{\abs{\alpha}\leq
  k}\sup_{x}\abs{D^{\alpha}f(x)}+\sum_{\abs{\alpha}=k}\sup_{x\neq
  y}\abs{f^{(k)}(x)-f^{(k)}(y)}/\norm{x-y}$ denote the Hölder norm, with the
convention that $\norm{f}_{C^{0}}=\sup_{x}\abs{f(x)}$, and that
$\norm{f}_{C^{k}}=\infty$ if $f$ is not $(k-1)$-times differentiable. We say
that $f$ has Hölder smoothness index $k$ if $\norm{f}_{C^{k}}<\infty$ \parencite[e.g.][Section 2.7.1]{vVWe96}.
This quantifies the ``smoothness'' of $f$ ($k$ is
also called the Hölder exponent; for simplicity we focus attention on exponents
that are integers). In other words, $f$ has smoothness $k$ if it is $k$ times
differentiable almost everywhere, with the derivatives bounded.

The next result shows that if (a) the conditional density of $f(e;x)$ of $e$ given $X=x$ is sufficiently smooth in the second argument, and (b) $g_{t}^{*}(x, e)$ is also sufficiently smooth in the second argument, then the smoothness of $g_{t}$ is given by the \emph{sum} of the smoothness indices of $x\mapsto g_{t}^{*}(x, e)$ and that of $e\mapsto f(e;x)$. This makes precise the notion that measurement error ``smooths out'' the non-linearities in $g_{t}^{*}$.

\begin{lemma}\label{lemma:smoothness}
  Suppose that $(X, e)$ has bounded support, and that the distribution of $e$
  given $X=x$ is continuous with bounded density $f(e; x)$ such that
  $\norm{\partial^{s} f/\partial x^{s}}_{C^{r}}<\infty$ for some non-negative
  integers $r, s$. Let $h(x, e)$ be a function such that
  $\sup_{e}\norm{\partial^{r} h(\cdot, e)/\partial e^{r}}_{C^{s}}<\infty$. Then
  $g(x):=E[h(X+e, e)\mid X=x]$ has smoothness $s+r$.
\end{lemma}
\begin{proof}
  Since the lower-order derivatives $D^{(0,k-1)}f(e;x)$ and
  $D^{(k-1,0)}h(x+e, e)$ exist and are Lipschitz continuous for $k\leq s$, by
  dominated convergence theorem, we can take a derivative under the integral
  sign using the Leibniz product rule, so that, for all $x$,
  \begin{equation*}
    g^{(s)}(x)=\int \sum_{k=0}^{s}\binom{s}{k} D^{(k,0)}h(x+e, e)D^{(0,s-k)}f(e; x)de.
 \end{equation*}
 Thus,
 \begin{multline*}
   g^{(r+s)}(x)=\sum_{k=0}^{s}\binom{s}{k}\frac{d^{r}}{dx^{r}}\int D^{(k,0)}h(y,
   y-x)D^{(0,s-k)}f(y-x; x)dy\\
   =\sum_{k=0}^{s}\sum_{u=0}^{r}\sum_{v=0}^{u}\binom{s}{k}\binom{r}{u}\binom{u}{v}(-1)^{r-u+v}\int
   D^{(k, r-u)}h(y, y-x)D^{(v, s-k+u-v)}f(y-x; x)dy,
 \end{multline*}
 where the first line follows by change of variables, and the second line by the dominated convergence theorem and Leibnitz product rule. Since $D^{(k, r-u)}h(y, y-x)$ and $D^{(v, s-k+u-v)}f(y-x; x)$ are bounded, it follows that $g^{(r+s)}$ is bounded.
\end{proof}

\subsection{Visualization of smoothness constants}\label{sec:empirical_extra}

Here we assess the smoothness constants suggested by the \acp{ROT} using the
visualization approach proposed in \textcite{NoRo21}. To explain the approach,
suppose that we are interested in a sharp \ac{RD} regression of an outcome
$\tilde{Y}_{i}$ on a running variable $\tilde{X}_{i}$, and make the assumption
that the conditional mean satisfies
$E[\tilde{Y}_{i}\mid \tilde{X}_{i}]\in\mathcal{F}_{RD}(M)$. To assess the
plausibility of the smoothness bound $M$, we regress $\tilde{Y}_{i}$ on a basis
function transformation $g(\tilde{X}_{i})$ of $\tilde{X}_{i}$, and the
interaction of $g(\tilde{X}_{i})$ with $\1{\tilde{X_{i}}\geq 0}$. To ensure that
the estimated regression function lies in $\mathcal{F}_{RD}(M)$, we minimize the
sum of squared residuals subject to the constraint that the second derivative of
the estimated regression function be no larger than $M$, and equal $M$ at the
cutoff. If the basis is sufficiently flexible, the estimated regression function
will tend to overfit the data, and therefore represent an extremal element of
$\mathcal{F}_{RD}(M)$. If the estimated regression function appears relatively
smooth, this thus is an indicator that the choice of $M$ is quite optimistic; if
we are clearly overfitting the data, it signals that the choice of $M$ is
conservative---it is unlikely that $E[\tilde{Y}_{i}\mid \tilde{X}_{i}]$ lies
outside $\mathcal{F}_{RD}(M)$.

\begin{figure}[t]
\centering
\input{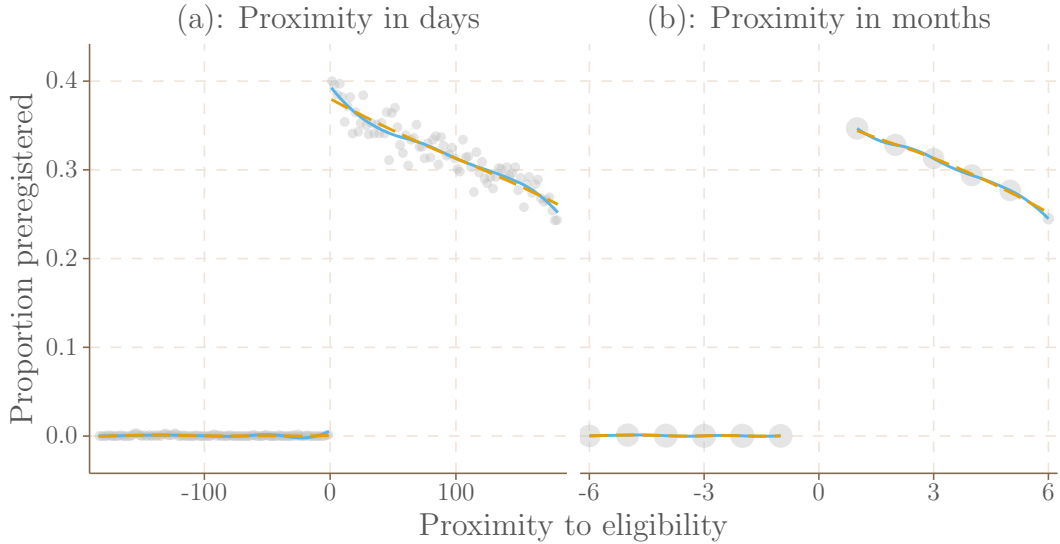}
\caption{Visualization of extreme conditional mean functions in the class
  $\mathcal{F}_{RD}(M_{t})$, for different choices of the first stage smoothness
  constant $M_{t}$.}\label{fig:a_fs}
\floatfoot{\emph{Notes:} Orange dotted line visualizes the choice of $M_{t}$
  based on the \ac{ROT} proposed by \textcite{ImWa19}. Blue solid line
  visualizes $M_{t}$ based on the \ac{ROT} proposed by \textcite{ArKo20}. Values
  of these constants are given in columns (4) and (5) of \Cref{tab:fs}. In panel
  (a), proximity is measured in days, and each point corresponds to an average
  of 1,000 individuals. In panel (b), proximity is measured in months, and each
  point corresponds to an average across all individuals born in a given month.}
\end{figure}

We use this method to assess the plausibility of the \acp{ROT} that we used to
calibrate the bounds $M_{y}$ and $M_{t}$ in the first stage and reduced form
sharp RD regressions. In the former, the outcome $\tilde{Y}_{i}$ corresponds to
the treatment variable $T_{i}$, while in the latter, $\tilde{Y}_{i}=Y_{i}$. To
implement the method, as a basis function, we use a quadratic spline with 21
knots on each side of the cutoff when proximity is measured in days, and with 6
knots when it is measured in months.

\begin{figure}[tp]
\input{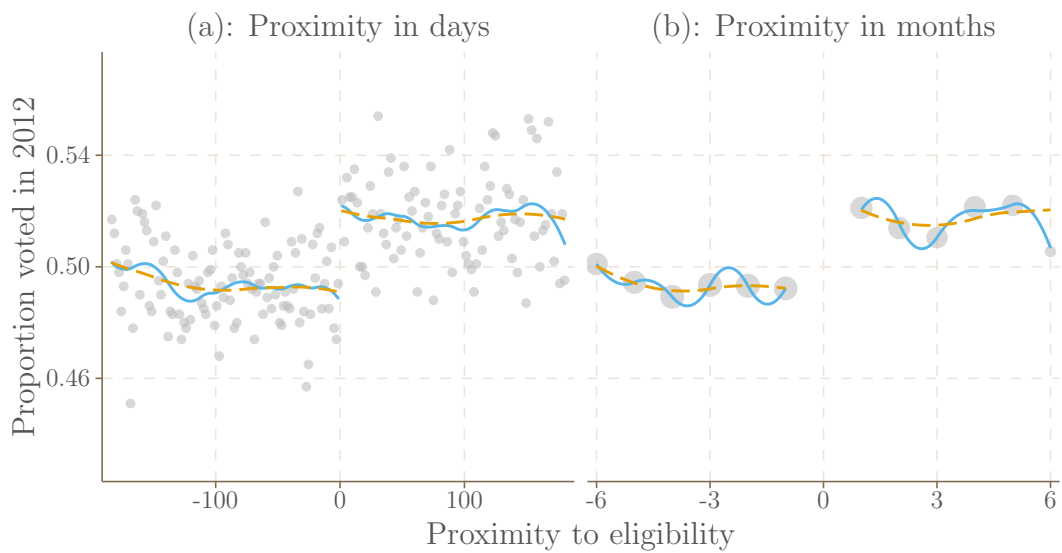}
\caption{Visualization of extreme conditional mean functions in the class
  $\mathcal{F}_{RD}(M_{y})$, for different choices of the reduced form
  smoothness constant $M_{y}$.}\label{fig:a_rf}
\floatfoot{\emph{Notes:} Orange dotted line visualizes the choice of $M_{y}$
  based on the \ac{ROT} proposed by \textcite{ImWa19}. Blue solid line
  visualizes $M_{y}$ based on the \ac{ROT} proposed by \textcite{ArKo20}. Values
  of these constants are given in columns (4) and (5) of \Cref{tab:rf}. In panel
  (a), proximity is measured in days, and each point corresponds to an average
  of 1,000 individuals. In panel (b), proximity is measured in months, and each
  point corresponds to an average across all individuals born in a given month.}
\end{figure}%

\Cref{fig:a_fs} visualizes the choices for the first stage smoothness constant
$M_{t}$, as estimated by the \acp{ROT} proposed by \textcite{ArKo20} and
\textcite{ImWa19}. Both choices \acp{ROT} appear reasonable based on the figure.
\Cref{fig:a_rf} gives an analogous visualization for the choices for the reduced
form smoothness constant $M_{y}$. Here the \ac{ROT} proposed by \textcite{ArKo20}
is quite conservative, while the \textcite{ImWa19} \ac{ROT} is more optimistic.

\subsection{Survey of empirical literature}\label{sec:survey}

To understand the prevalence of measurement error issues in the applied papers
that use \ac{RD} designs, we surveyed articles in 7 leading journals
(\emph{American Economic Journal: Applied Economics, %
  American Economic Journal: Economic Policy, %
  American Economic Review, %
  Quarterly Journal of Economics, %
  Journal of Political Economy}, %
\emph{Review of Economics and Statistics}, and \emph{Review of Economic
  Studies}) published between 2005 and 2020. We identified 139 papers that used
RD design, of which 32 papers (23\%) featured a running variable measured with
error. For each of these articles, we classified the type of measurement error,
whether Assumption~\ref{item:c1} holds (i) outright, (ii) after doughnut
trimming; or, under heaping error, (iii) after dropping the heaping points.
Finally, we noted the type of correction employed in the paper.

\Cref{tab:survey} reports the survey results. 27 out of the 32 papers (84\%)
feature rounding or grouping error. Out of these, 17 require no explicit
measurement error correction, provided we interpret the estimand correctly. 3
papers correctly employ doughnut trimming (or else include a dummy for the
cutoff month, which has the same effect). Several papers don't quite deal with
measurement error issues correctly, either by failing to create a doughnut hole,
or else by accounting for the discreteness of the rounded running variable by
clustering the standard errors by the running variable.\footnote{As argued in
\textcite{KoRo18}, clustering exacerbates, rather than solves, any inference
issues that the discreteness causes.} No paper discusses the implications of the
measurement error for the interpretation of the estimand.

Our survey indicates that even though grouping error is fairly common in
practice, there is a lack of clarity among applied researchers in how to account
for it. Since our survey focuses on the most selective journals, it is likely
that due to selection bias, the prevalence of measurement error is even higher
than our 23\% estimate. Likewise, while grouping error is very common, it likely
accounts for a lower share of measurement error types than 84\%, as we find in
our survey, since other types of measurement error are more difficult to deal
with.

\begin{table}[t]
  \centering
  \caption{Survey of empirical \ac{RD} papers with measurement error in the
    running variable.}\label{tab:survey}
  {\scriptsize
  \begin{tabular}{@{}p{0.28\textwidth}p{0.28\textwidth} cc p{0.22\textwidth}@{}}
    & &\ref{item:c1}& Error&  \\
    Paper& Running variable& holds& type& Correction\\
    \midrule
  Snyder and Evans (2006)& quarter of birth& Y& R& Not needed \\
    Black et al.\ (2007)& rounded risk score & Y& R& Not needed\\
    Card et al.\ (2007)& job tenure in months& Y& R& Not needed\\
  Anderson et al.\ (2012)& age in months & Y & R & Not needed\\
    Magruder (2012) & distance from magisterial district to
                      bargaining council regime border & Y & R & Not needed\\
    Clark and Royer (2013)& month-year of birth & Y& R& Not needed\\
    Borghans et al.\ (2014) & month-year of birth & Y & R & Not needed\\
  Anderson et al.\ (2014)& age in months& Y & R & Not needed \\ %
    Kumar (2018) & distance from county centroid to Texas border & Y & R & Not needed \\
  Avdic and Karimi (2018)& month-year of birth & Y& R& Not needed \\
    Wherry et al.\ (2018) & month-year of birth & Y & R & Not needed\\
    Malamud and Pop-Eleches (2010) & month-year of birth & Y & R & CRV\\
    Chetty et al.\ (2014) & income bins & Y & R & CRV\\
  Haggag and Paci (2014)& interval-reported taxi fare & Y& R& CRV \\
  Erten and Keskin (2018)& month-year of birth& Y& R& CRV\\
  Erten and Keskin (2020)& month-year of birth & Y& R& CRV \\
  Dieterle et al.\ (2020)& distance from county centroid to state border &
                                                                                 Y
                                 & R & \textcite{BaBrDi20} \\

  Lalive (2007) & month-year of birth & YD & R&~\\
  Stancanelli and van Soest (2012)& month-year of birth& YD& R&~\\

    Li et al.\ (2015)& age in years& D & R& Doughnut\\
Carpenter and Dobkin (2017)& age in months& D& R&Dummy for cutoff month \\
Kreiner et al.\ (2020)& age in months & D& R & Dummy for cutoff month, CRV\\
  Oreopoulos (2006)& year of birth& D & R & CRV \\
  Johnston and Mas (2018)& week of unemployment insurance claim & D & R&\\
  Davis (2008)& calendar month or year& D& R&~\\
  Lleras-Muney (2005)& year of birth, sometimes reported in multiples of 10 & YD/H& R/H&~\\
  Almond et al.\ (2011)& birth weight& H& H & Series of doughnuts \\
  Barreca et al.\ (2011)& birth weight& H& H & Series of doughnuts \\
  Almond et al.\ (2010) & birth weight &H& H &~\\

  Gonz\'{a}lez (2013) & month of birth / month of abortion / estimated month of conception
                       & Y/D/D & R/R/O & Discussed possible attenuation bias when using estimated month of conception\\

    Battistin et al. (2009)& age in years & N& O & Not needed (under maintained assumptions)\\
    Becker et al.\ (2013) & revised regional GDP per capita & N & O &\\
\bottomrule
  \end{tabular}}
  \floatfoot{\emph{Notes:} Whether Assumption~\ref{item:c1} holds is coded as:
    Y---Yes; D---Yes, after doughnut trimming; YD---Yes, possibly after doughnut
    trimming (there is insufficient information in the paper to determine
    whether doughnut trimming is needed); H---Yes after dropping heaping points;
    N---No. Error type is coded as: R---Rounding or grouping error; H---Heaping;
    O---Other. CRV\@: clustered standard errors by running variable following
    \textcite{LeCa08}.}
\end{table}

\end{appendices}

\end{document}